\documentclass[12pt,english]{article}

\usepackage[utf8]{inputenc}
\usepackage[T1]{fontenc}
\usepackage{babel}
\usepackage[round]{natbib}

\usepackage{amsmath,amssymb}
\usepackage{amsthm}
\usepackage{txfonts}
\DeclareSymbolFont{symbols}{OMS}{cmsy}{m}{n}

\usepackage{booktabs}
\newcommand{\otoprule}{\midrule[\heavyrulewidth]}

\usepackage{tabularx}
\newcolumntype{W}{>{\centering\arraybackslash}X}
\newcolumntype{C}{>{\centering\arraybackslash}X}
\usepackage{multirow}

\usepackage[algoruled,linesnumbered,noend]{algorithm2e}

\SetCommentSty{mycommentsty}

\usepackage{textcomp}\usepackage{xcolor}

\usepackage{tikz}
\usetikzlibrary{shapes,arrows}
\usetikzlibrary{snakes}
\usetikzlibrary{positioning,patterns}
\usepackage{verbatimbox}

\usepackage{array}
\usepackage{url}
\usepackage{xspace}

\newcommand{\ie}{\textit{i.e.},\xspace}

\renewcommand{\l}{\ensuremath{\ell}}
\renewcommand{\emptyset}{\ensuremath{\varnothing}}

\newcommand{\SA}{\ensuremath{\mathit{SA}}}
\newcommand{\LCP}{\ensuremath{\mathit{L}}}
\newcommand{\Occ}{\ensuremath{\mathit{Occ}}}

\newcommand{\suff}[1]{\textsf{\textup{suff}\ensuremath{(#1)}}\xspace}
\newcommand{\pref}[1]{\textsf{\textup{pref}\ensuremath{(#1)}}\xspace}
\newcommand{\substr}[1]{\textsf{\textup{substr}\ensuremath{(#1)}}\xspace}
\newcommand{\listpref}[1]{\textsf{\textup{listpref}\ensuremath{(#1)}}}
\newcommand{\q}[1]{\textsf{\textup{q}\ensuremath{(#1)}}\xspace}

\newtheorem{theorem}{Theorem}
\newtheorem{lemma}[theorem]{Lemma}
\newtheorem{proposition}[theorem]{Proposition}

\newtheorem{corollary}[theorem]{Corollary}

\newtheorem{definition}{Definition}

\usepackage{mdwtab}

\title{FSG: Fast String Graph Construction for De Novo Assembly}

\author{}
\date{}

\begin{document}

\maketitle

\begin{small}
\noindent
\textbf{Authors}\\[1ex]
\begin{tabular}{lp{1ex}ll}
&& E-Mail & Telephone\\
Paola Bonizzoni$^{\star}$ &&
\url{bonizzoni@disco.unimib.it} & (+39)02~6448~7814\\
Gianluca Della Vedova &&
\url{gianluca.dellavedova@unimib.it} & (+39)02~6448~7881\\
Yuri Pirola &&
\url{pirola@disco.unimib.it} & \\
Marco Previtali &&
\url{marco.previtali@disco.unimib.it} & (+39)02~6448~7917\\
Raffaella Rizzi &&
\url{rizzi@disco.unimib.it} & (+39)02~6448~7838\\
\end{tabular}

\bigskip
\noindent
\textbf{Address:}\\[1ex]
\begin{tabular}{l}
 Dipartimento di Informatica Sistemistica e Comunicazione,\\
 Universit\`a degli Studi di Milano-Bicocca,\\
 Viale Sarca 336, 20126 Milan, Italy.\\
\end{tabular}
\end{small}

\clearpage

\begin{abstract}
The string graph for a collection of next-generation  reads  is a lossless data
representation that is fundamental for de novo assemblers based on the overlap-layout-consensus paradigm.
In this paper, we explore  a novel approach to compute the string
graph, based on the FM-index and Burrows-Wheeler Transform (BWT). 
We describe   a simple algorithm  that uses only the FM-index representation of the collection of 
reads to construct the string graph, without accessing the input reads.
Our algorithm has been integrated into the SGA assembler as a standalone module to
construct the string graph.    

The new integrated assembler has been assessed on a standard
benchmark, showing that FSG is significantly faster than SGA while maintaining a moderate use of main memory, and
showing practical advantages in running FSG  on  multiple threads.
Moreover, we have studied the effect of coverage rates on the running times.    
\end{abstract}

\textbf{Keywords:} string graph; Burrows--Wheeler Transform; genome assembly

\section{Introduction}
\label{sec:introduction}

De
novo sequence assembly continues to be one of the most fundamental
problems in Bioinformatics.    
Most of the available
assemblers~\citep{Simpson2009,Peng2010,bankevich2012spades/long,DBLP:journals/almob/ChikhiR13,DBLP:journals/almob/SalikhovSK14,DBLP:journals/jcb/ChikhiLJSM15}
are based on the notions of de Bruijn graphs
and $k$-mers (short $k$-long substrings of input data).
Currently, biological data are produced by different Next-Generation Sequencing
(NGS) technologies which routinely and cheaply produce a large number of reads whose
length varies according to the specific technology.    
For example, reads obtained by Illumina technology (which is the most used) usually have length
between 50 and 150 bases~\citep{salzberg2012gage}.

To analyze datasets coming from different technologies, hence with a
large variation of read lengths, an approach based on same-length strings is
likely to be limiting, as witnessed by the recent introduction of variable order de
Bruijn graphs~\citep{boucher2015variable}.
The \emph{string graph}~\citep{Myers2005} is an alternative approach that
does not need to break the reads into k-mers (as in the de Bruijn graphs), with the
advantage of immediately distinguishing repeats longer than $k$ but contained in a
read---when using a de Bruijn graph, those repeats are resolved only at later
stages.
The string graph is the main data representation used by assemblers based on the  overlap-layout-consensus paradigm.
Indeed, in a string graph the vertices are the input reads and the arcs corresponds to
overlapping reads, with the property that contigs are paths of the
string graph.    

Even without repetitions, analyzing only $k$-mers instead of the
full-length reads can result in some information loss, since two bases that are $k+1$
positions apart can belong to the same read, but are certainly not part of the same
$k$-mer.
Indeed, differently from de Brujin graphs,  any path of a string graph  is a valid assembly of reads.
String graphs are more computationally intensive to
compute~\citep{Simpson2012}, justifying our search for faster algorithms.    
The most widely used string graph assembler is SGA~\citep{Simpson2010}, which
first constructs the BWT~\citep{Burrows1994} and the FM-index of
a set of reads, and then uses those data structures to
efficiently compute the arcs of the string graph (connecting overlapping reads).
Another string graph assembler is Fermi~\citep{li_exploring_2012} which
implements a variant  of the original  SGA algorithm~\citep{Simpson2010} that is  tailored for SNPs and  variant calling.    

Several recent papers face the problem of designing efficient
algorithmic strategies or data structures for building string graphs.
Among those works, we can find a string graph
assembler~\citep{Ben-Bassat15122014}, based on a careful use of
hashing and Bloom filters, with performance comparable with the first
SGA implementation~\citep{Simpson2010}.
Another important alternative approach to SGA is
Readjoiner~\citep{readjoiner} which is based on an efficient
computation of  a subset of exact suffix-prefix  matches, and by subsequent rounds of
suffix sorting, scanning, and filtering, 
obtains the non-redundant arcs of the graph.

All currently available assemblers based on string graphs (such as SGA) need to both (1) query an
indexing data 
structures (such as an FM-index), and (2) access the original reads to detect prefix-suffix overlaps
between the elements. Since the self-indexing data structures, such as FM-index,
represent the whole information of the original dataset, an interesting problem
is to design efficient algorithms for the construction of string graphs that
only require to keep the index, while discarding the original reads. 

Improvements in
this direction have both theoretical and practical motivations. 
Indeed,
detecting prefix-suffix overlaps only by analyzing the (compressed) index 
is an almost unexplored problem.    

The information contained in the indexing data structure can be analyzed with
different and almost orthogonal approaches. 
A natural and straightforward goal that we have explored previously
is to minimize the amount of data maintained in RAM~\citep{Bonizzoni2015}.    
In this paper we will focus instead on reducing the running time, by
introducing a method that is able to build the whole string graph, via a limited
number of sequential scans of the index. This property leads to the design of an
algorithm that can exploit some features of modern processors.    
Moreover, since our new algorithm computes the string graph, we have a memory conscious and time efficient tool
that may be directly integrated in a pipeline for assembling DNA reads.    

We propose a new algorithm, called FSG, to compute the string graph of a set $R$ of
reads with $O(nm^{m})$ worst-case time complexity --- $n$ is the number of reads in
$R$ and $m$ is the maximum read length.    
To the best of our knowledge, it is the first algorithm that computes a string graph using
only the FM-index of the input reads. 
The vast literature on BWT and FM-index hints that this approach is amenable to further research.    
Our algorithm is based on a characterization of the string graph given
in~\citep{Bonizzoni2015}, but we follow a completely
different approach.    

An important observation is that, to compute the arcs outgoing from each read $r$,  
SGA  queries the FM-index for each character of $r$.
While this  approach works in linear, \ie $O(nm)$, time, it can perform several
redundant queries, most notably when the reads share common suffixes
(a very common case).
Our algorithm queries the FM-index in a specific order, so that each distinct
string is processed only once, while SGA might process more than once
each repeated string.   

It is important to notice that our novel algorithm uses 
a characterization of a string graph~\citep{DBLP:conf/wabi/BonizzoniVPPR14} that is
different, but equivalent, to the one in~\citep{Myers2005}.
We have  implemented  FSG and integrated it with the SGA assembler, by replacing the
procedure to construct the string graph. 
Our implementation follows the SGA guidelines, that is we use SGA's read correction step
before computing the overlaps without allowing mismatches (which is also  SGA's  default
choice).    
Indeed, the guidelines to reproduce the assembly of the dataset
NA12878 included in the SGA software package set the \texttt{--error-rate}
parameter to $0$, the default value. Therefore, it is fair to compare the
performances of the two tools.
Also the assembly phases of SGA can be applied without any modification.    
These facts guarantees that the assemblies produced by our approach and
SGA are the same, except for the unusual case when two reads have two different overlaps.
In that case, SGA considers only the longer overlap, while we retain all overlaps.    
While it is trivial to modify our approach to guarantee the the assembly is the same, we
have decided that considering all overlaps is more informative.    
We want to point out that the FSG algorithm is relatively simple and could be useful also
for different assembly strategies.


We have compared FSG with SGA ---
a finely tuned implementation that has
performed very nicely in the latest Assemblathon
competition~\citep{bradnam2013assemblathon} ---
where we have used the latter's default parameter (that is,
we compute overlaps without errors).    
Our experimental evaluation on a
standard benchmark dataset shows that our approach is 2.3--4.8 times faster than SGA
in terms of wall clock time (1.9--3 times in terms of user time).
requiring only
2.2 times more memory than SGA.


\section{Preliminaries}
\label{sec:prelim}

We  briefly recall some standard definitions that will be used in the following.
Let $\Sigma$ be a constant-sized alphabet and let $S$ be a string over $\Sigma$.
We denote by $S[i]$ the $i$-th symbol of $S$, by $\l=|S|$ the length of $S$, and by
$S[i:j]$ the substring $S[i]S[i+1] \cdots  S[j]$ of $S$.
The \emph{suffix} and \emph{prefix} of $S$ of length $k$ are the
substrings $S[\l-k +1: \l]$ (denoted by $S[\l-k +1:]$) and $S[1: k]$ (denoted by
$S[:k]$) respectively.
Given two strings $(S_i, S_j)$, we say that $S_i$
\emph{overlaps} $S_j$  iff a nonempty suffix $\beta$ of $S_i$ is
also a prefix of $S_j$, that is $S_{i}=\alpha\beta$ and $S_j = \beta\gamma$.
In that case we say that that $\beta$ is the \emph{overlap} of $S_i$ and $S_j$,  denoted as $ov_{i,j}$,
that $\gamma$ is the \emph{right extension} of $S_i$ with $S_j$, denoted as
$rx_{i,j}$, and $\alpha$ is
the  \emph{left extension}  of  $S_j$ with $S_i$, denoted as
$lx_{i,j}$.

In this paper we consider a set $R$ of $n$ strings over $\Sigma$ that are terminated by
the sentinel $\$$, which is the smallest character. 
To simplify the exposition, we will assume that all input strings have exactly $m$ characters, excluding the $\$$.
The \emph{overlap  graph} of a set $R$ of strings is  the directed
graph $G_O = (R,  A)$ whose vertices are the strings in $R$.    
For each three strings $\alpha$, $\beta$, and $\gamma$ such that 
$r_i= \alpha\beta$ and $r_j= \beta\gamma$ are two strings, there is 
the arc $(r_i, r_j) \in A$. 
In this case $\beta$ is called the \emph{overlap} of the arc. 

Observe that the notion of overlap graph originally given by~\citep{Myers2005} is defined
by labeling 
with the \emph{right extension} $rx_{i,j}=\gamma$ the arc $(r_i, r_j) \in A$. The assembly string related to $(r_i, r_j)$ is given by $r_i \gamma$. More in general, given a path $\pi= \langle r_1, r_2, \cdots, r_n \rangle$ in the overlap graph, the assembly string is $r_1 rx_{1,2} \cdots, rx_{n-1,n}$.

The  notion of a string graph derives from the observation that in a overlap graph
the label of an arc $(r,s)$ may be equal to the assembly string of a path $\langle r,
\ldots, s \rangle$: in this case the arc $(r,s)$ is called called \emph{redundant} and it
can be removed from the overlap graph 
without loss of information, since all paths resulting in a valid assembly are still in
the graph, even after the removal of such redundant arcs $(r,s)$.
In~\citep{Myers2005} redundant arcs are those arcs $(r,s)$ labeled by $\gamma$, for
$\gamma$ containing as prefix the label of an arc $(r,t)$.
In~\citep{bonizzoni16:_exter_memor_algor_strin_graph_const} we state an equivalent
characterization of string graphs (given below) which is a direct consequence of the fact that an arc $(r_i, r_j)$ is labeled by the \emph{left extension} $\alpha$ and its assembly is $\alpha r_j$.
An arc $e_1 = (r_i, r_j)$ of $G_{O}$ labeled by $\alpha$ is
\emph{transitive} (or \emph{reducible}) if there exists another arc $e_2 = (r_k, r_j)$
labeled by   $\delta$ where $\delta$ is a suffix of $\alpha$.
Therefore, we say that $e_1$ is \emph{non-transitive} (or \emph{irreducible}) if no such arc $e_{2}$ exists.
The string graph of $R$ is obtained from $G_{O}$ by removing all
reducible arcs.
This definition allows to use directly the FM-index to compute the labels of the overlap
graph since the labels are obtained  by 
backward extensions on the index.

The \emph{Generalized Suffix Array (GSA)}~\citep{Shi1996} of $R$ is the $n$-long array
$\SA$ where each element $\SA[i]$ is equal to $(k, j)$ if and only if the $k$-long suffix $r_{j}[|r_j|-k+1:]$ of the string $r_{j}$ is the
$i$-th smallest element in the lexicographic ordered set of all suffixes of the
strings in  $R$.
The \emph{Longest Common Prefix (LCP) array} of  $R$, is the $n$-long array
$\LCP$  such that $\LCP[i]$  is equal to
the length of the longest prefix shared by the the $k_{i}$-suffix of
$r_{j_{i}}$ and the  $k_{i-1}$-suffix of $r_{j_{i-1}}$, where $\SA[i]=(k_{i}, j_{i})$ and  $\SA[i-1]=(k_{i-1}, j_{i-1})$.
Conventionally, $\LCP[1]=-1$.

The \emph{Burrows-Wheeler Transform (BWT)} of $R$ is the sequence $B$ such that
$B[i]=r_{j}[|r_j|-k]$, if $\SA[i] = (k,j)$ and $k > 1$, or $B[i]= \$$, otherwise.
Informally, $B[i]$ is the symbol that precedes the $k$-long suffix of
a string $r_j$ where such suffix is the $i$-th smallest suffix in the
ordering given by $\SA$.

The $i$-th smallest (in lexicographic order) suffix  is denoted by $LS[i]$, that is
if $SA[i]=(k,j)$ then $LS[i]= r_{j}[|r_{j}| - k + 1:]$.
Given a string $\omega$, all suffixes of $R$ whose prefix is $\omega$ appear consecutively in $LS$.
We call \emph{$\omega$-interval}~\citep{bauer_lightweight_2013} the
maximal interval $[b, e]$ such
that $\omega$ is a prefix of $LS[i]$ for each $i$, $b\le i\le e$.
We denote the $\omega$-interval by $\q{\omega}$.
The \emph{width} $e-b+1$ of the $\omega$-interval is equal to the number of occurrences of
$\omega$ in some read of $R$.
Since $LS$, the BWT $B$ and $\SA$ are closely related, we also say that $[b, e]$ is a
$\omega$-interval on all those arrays. 
Given a $\omega$-interval and a character $c$, the \emph{backward
  $c$-extension} of the $\omega$-interval is the (possibly empty) $c \omega$-interval.
We recall that the FM-index~\citep{Ferragina2005} is essentially made of the
two functions $C$ and $\Occ$, where
$C(c)$, with $c$ a character, is the number of occurrences in $B$ of
characters that are alphabetically smaller than $c$, while
$\Occ(c,i)$ is the number of occurrences of $c$ in the prefix $B[ : i - 1]$.
Given a string $\alpha$ and a character $c$,
the backward $c$-extension of  $\q{\omega}= [b, e]$ is
$q(c\omega) = [C(c) + Occ(c, b) + 1, C(c) + \Occ(c, e+1)]$~\citep{Ferragina2005}.

\section{The Algorithm}
\label{sec:algorithm}

Our algorithm is based on two steps: the first is to compute the
overlap graph, the second is to remove all transitive arcs.
Given a string $\omega$ and $R$ a set of strings (reads), let $R^S(\omega)$ and $R^P(\omega)$ be
respectively the subset of $R$ with suffix (resp.\ prefix) $\omega$.
As usual in string graph construction algorithms, we   will assume that the set $R$ is \emph{substring free}, that is no string is a substring of another one. 
A fundamental observation is that the list of all
nonempty overlaps $\beta$ is a compact representation
of the overlap graph, since all pairs in $R^S(\beta) \times
R^P(\beta)$ are arcs of the overlap graph.
Moreover, each arc $(r_{i}=\alpha\beta, r_j = \beta\gamma)$ of the overlap graph can be
represented by the triple $(\alpha, \beta, \gamma)$.    

Our approach to compute all overlaps between pairs of strings
is based on the notion of \emph{potential overlap}, which
is a nonempty string $\beta^{*}\in \Sigma^+$ that is a proper suffix of an input string
$r_i = \alpha \beta^{*}$ ($\alpha \ne \epsilon$) and such that there exists an input
string $r_j=\gamma\beta^{*}\delta$, with $\delta \ne \epsilon$, containing $\beta^{*}$ as a
substring, but not a suffix.    

A simple relation between overlaps and potential overlaps is given in
Proposition~\ref{proposition:po-extension}, which is a direct consequence of the
definition of potential overlap.

\begin{proposition}\label{proposition:po-extension}
Let $\beta$ be an overlap.
Then all suffixes of $\beta$ are potential overlaps.
\end{proposition}

We can now briefly sketch our algorithm which consists of two main parts.    
The first part computes all
potential overlaps, starting from those of length $1$ and extending
the potential overlaps by adding a new leading character.
Each potential overlap is also checked to determine whether it is also
an actual overlap, to compute the set of all overlaps.    

The second part of our algorithm, that is to detect all transitive
arcs, starts from the sets  $ARC(\alpha = \epsilon, \epsilon\beta, X=R^p(\beta))$ (a set for each one of the overlaps $\beta$) that
can be immediately obtained from the overlaps $\beta$ computed in the first step,
where in general a set $ARC(\alpha , \alpha\beta, X)$ consists of arcs with overlap $\beta$, a left extension that has $\alpha$ as a suffix, and are incoming into a read in $X$. Observe that $ARC(\alpha = \epsilon, \alpha\beta, X=R^p(\beta))$ is the set of the arcs of the overlap graph having overlap $\beta$.
During the second part of the algorithm, the transitive arcs of the overlap graph are
removed by iteratively computing arc-sets of increasing extension length $\ell$ (starting
from $\alpha = \epsilon$) by adding a leading character to the extensions $\alpha$ of
length $\ell-1$ (of the sets computed at same previous iteration) and deleting reads from
$X$. All the computed sets $ARC(\alpha , \alpha\beta, X)$, where $\alpha\beta$ is a read
(that is, $\alpha$ is the complete left extension of the arcs), contains only irreducible arcs.
The following definition formalizes the previously discussed notions.    

\begin{definition}
\label{definition:arc-triple}
Let $\alpha$ be a string, let $\beta$ be a nonempty string, and let $X$ be a subset of  $R^P(\beta)$.
The \emph{arc-set}  $ARC(\alpha, \alpha\beta,  X)$ is the set $\{ (r_{1}, r_{2}) :
\alpha\beta\text{ is a suffix of }r_{1},
\beta\text{ is a prefix of } r_{2} \text{, and } r_{1}\in R, r_{2}\in
X\}$.
The strings $\alpha$ and $\beta$ are called the \emph{extension} and the
\emph{overlap} of the arc-set.
The set $X$ is called the \emph{destination} set of the arc-set.
\end{definition}

An arc-set $ARC(\alpha, \alpha\beta,  X)$ is \emph{terminal} if
there  exists $r \in R$ s.t. $r = \alpha\beta$, while an arc-set is \emph{basic} if $\alpha = \epsilon$  (that is the empty string).
Since the arc-set $ARC(\alpha, \alpha\beta,  X)$ is uniquely
determined by strings $\alpha$, $\alpha\beta$, and $X$, the triple $(\alpha,
\alpha\beta,  X)$ wil be used in our algorithm to encode the arc-set $ARC(\alpha, \alpha\beta,  X)$.

Observe that the string $\alpha$ in Definition~\ref{definition:arc-triple} is a suffix of the left extension (label) of the arcs in $ARC(\alpha, \alpha\beta,X)$. When the arc-set is \emph{terminal} then the extension $\alpha$ of the arc-set is also the label of its arcs.
   
Moreover, the arc-set $ARC(\alpha, \alpha\beta, X)$ is 
\emph{correct} if $X\supseteq \{ r_{2}\in R^P(\beta): r_{1}\in
R^S(\alpha\beta)\text{ and }  (r_{1}, r_{2}) \text{ is irreducible}\}$, that is all
irreducible arcs   whose overlap is $\beta$ and whose (left) extension has a suffix $\alpha$ are incoming in a read of $X$.

Observe that our algorithm outputs only correct arc-sets (hence all irreducible arcs are preserved).    
Moreover terminal arc-sets, computed by our algorithm, only contain irreducible arcs (see Lemma
\ref{lemma:all-triples-are-correct}), hence all transitive arcs are removed.





\begin{definition}
\label{definition:transitive-arc}
An arc $(r_1, r_3)$ is \emph{transitive}
iff there exists an arc $(r_2, r_3)$ whose extension is a suffix of the extension of
$(r_1, r_3)$~\citep{Bonizzoni2015}.
\end{definition}

Our algorithm is based on an extension of the definition of transitive arc to arc-sets.
We present Algorithm~\ref{algorithm:sketch-overlap} which computes the overlap graph of a
set of input strings (the overlap graph is represented by the set of all overlaps), and
Algorithm~\ref{algorithm:sketch-reduction} which receives the overlaps computed by
Algorithm~\ref{algorithm:sketch-overlap} and outputs the string graph.
In our description we assume that, given a string
$\omega$, we can compute in constant time
(1) the number \suff{\omega} of input strings whose suffix is $\omega$,
(2) the number \pref{\omega} of input strings whose prefix is $\omega$,
(3) the number \substr{\omega} of occurrences of $\omega$ in the input strings.
Moreover, we assume to be able to list the set $\listpref{\omega}$ of input
strings with prefix $\omega$ in  $O(|\listpref{\omega}|)$ time.
In the next section we will describe how to compute such a data structure.

We recall that Algorithm~\ref{algorithm:sketch-overlap}
exploits Proposition~\ref{proposition:po-extension} to compute all overlaps. 
More precisely, given a $k$-long potential overlap $\beta^*$, the $(k+1)$-long string
$c\beta^*$, for $c \in \Sigma$, is a potential overlap if and only if $\suff{c\beta^{*}}>
0$ and $\substr{c\beta^{*}} > \suff{c\beta^{*}}$.
We construct incrementally all potential overlaps, first by
determining if each string  $\beta^*$ consisting of a single character is a potential
overlap.
Then, starting from the potential overlaps of length $1$, we
iteratively compute the potential overlaps of increasing length by
prepending each character $c\in \Sigma$ to each 
$k$-long potential overlap $\beta^{*}$ (stored in the list \emph{Last}), and we determine if $c\beta^{*}$ is an
$(k+1)$-long potential overlap: in this case we store the
potential overlap in the list \emph{New}.

The lists \emph{Last} and \emph{New} store the potential overlaps computed at the previous and at the current iteration respectively. 
By our previous observation, this procedure computes all potential overlaps.
Observe that a potential overlap $\beta^{*}$ is an overlap iff $\pref{\beta^{*}}
> 0$.
Since each potential overlap is a suffix of some input string,  there
are at most $nm$ distinct suffixes, 
where $m$ and $n$ are the length and the number of input strings, respectively.    
Each query \suff{\cdot}, \pref{\cdot}, \substr{\cdot} requires
$O(1)$ time, thus the time complexity of all of  such  queries  is $O(nm)$.
Given two distinct strings $\beta_1$ and $\beta_2$,
when $|\beta_{1}|=|\beta_{2}|$ then no input string can be
in both $\listpref{\beta_{1}}$ and  $\listpref{\beta_{2}}$.    
Since each overlap is at most $m$ long,  the overall
time spent in the \listpref{\cdot} queries is $O(nm)$.
The first phase produces (line~\ref{alg:0:listpref}) the set of disjoint \emph{basic}
arc-sets $ARC(\epsilon, \epsilon\beta, R^p(\beta))$ for each overlap $\beta$, whose union is exactly the  set of arcs of the overlap graph. Recall that $\listpref{\beta}$ gives the
set of reads with prefix $\beta$, which has been denoted by $R^p(\beta)$. 

The following definitions and lemma are fundamental in the design of Algorithm~\ref{algorithm:sketch-reduction}, since they allow to restrict the search that determines whether a certain arc is transitive. We recall that each arc-set $ARC(\alpha,\alpha\beta,X)$ is actually encoded as the triple $(\alpha,\alpha\beta,X)$.

\begin{definition}
\label{definition:reducing-arc-set}
Let $A=ARC(\alpha\gamma, \alpha\gamma\beta, X_1)$, $B=ARC(\gamma, \gamma\beta\delta, X_2)$
be two arc-sets, such that $B$ is terminal and $X_2 \subseteq X_1$.    
Then  $B$
  \emph{reduces} $A$ and the tuple
  $(\alpha\gamma, \alpha\gamma\beta, X_1 \setminus X_2)$ is the
  \emph{residual} arc-set of $A$ with respect to $B$, denoted by $A \setminus B$.

\end{definition}

Based on the previous definition, we say that the arcs of $(\alpha\gamma, \alpha\gamma\beta, X_1 \cap  X_2)$ are \emph{removed} by $B$, or $B$ \emph{removes} those arcs.

In Figure~\ref{fig:example-arc-set-1} and Table~\ref{fig:example-arc-set} an example of arc-sets is presented for a set of five reads.

\begin{figure}[htbp] 
  \caption{Example of overlap graph on five reads (on the left).
    Each arc is labeled with the corresponding left extension.
    Arcs with the same colored have the same overlap.
    Dashed arcs are transitive.}
  \label{fig:example-arc-set-1} 
\end{figure}

\begin{table}[ht!]
  \caption{Example of arc-sets corresponding to the overlap graph of Figure~\ref{fig:example-arc-set-1}.
    Among the five reported arc-sets, $B$, $C$ and $E$ are \emph{terminal}, and only $D$
    is not \emph{correct} (and non-terminal) since it represents reducible arcs of the
    overlap graph.
    Observe that $E$ reduces $D$ (by Definition~\ref{definition:reducing-arc-set}) and the
    residual arc-set $D \setminus E$ represents an empty set of arcs (all the four reducible
    arcs are removed).}\label{fig:example-arc-set} 
  \begin{center}
    \begin{tabular}{|lcc|}
    \hline
    Arc-set&Extension,Overlap&Represented arcs\\\hline
    $A$=(cg, cgtaca, $\{r_3\}$)&cg, taca&$(r_1,r_3), (r_2,r_3)$\\
    $B$=(ccg, ccgtaca, $\{r_3\}$)&ccg, taca&$(r_1,r_3)$\\
    $C$=(tcg, tcgtaca, $\{r_3\}$)&tcg, taca&$(r_2,r_3)$\\
    $D$=(ta, taca, $\{r_4,r_5\}$)&ta, ca&$(r_1,r_4), (r_1,r_5), (r_2,r_4), (r_2,r_5)$\\
    $E$=(ta, tacatgt, $\{r_4,r_5\}$)&ta, catgt&$(r_3,r_4), (r_3,r_5)$\\
    \hline
    \end{tabular}
  \end{center}  
\end{table}



\begin{lemma}
\label{lemma:reduction-by-terminal}
Let $(r_{1}, r_{2})$ be an arc with overlap $\beta$.
Then  $(r_1, r_2)$ is transitive if and only if (i) there exist
$\alpha, \gamma, \delta, \eta \in \Sigma^+$ such that $r_1 = \gamma\alpha\beta$,
$r_2 = \beta\delta\eta$, and (ii) there exists an input read
$r_{3}=\alpha\beta\delta$ such that $(r_{3}, r_{2})$ is an irreducible arc of a nonempty
correct arc-set $ARC(\alpha, \alpha\beta\delta, X)$ for all $X$ such that $r_{2}\in X$.
\end{lemma}

\begin{proof}
Let $r_{3}=\alpha\beta\delta$ be the input string maximizing $|\delta|$ so
that $r_1 = \gamma\alpha\beta$, $r_2 = \beta\delta\eta$, for some strings
$\alpha, \gamma, \delta, \eta \in \Sigma^+$.    
Notice that such $r_{3}$ exists iff $(r_{1}, r_{2})$ is transitive.
If no such input string $r_{3}$ exists,
then all arc-sets $ARC(\alpha, \alpha\beta\delta, \cdot)$ are empty.

Assume now that such an input string $r_{3}$ exists, we will prove that the arc $(r_{3},
r_{2})$ reduces $(r_{1}, r_{2})$.
First we prove that $(r_{3}, r_{2})$ is irreducible.
Assume to the contrary that $(r_{3}, r_{2})$ is transitive, hence
there exists an arc $(r_{4}, r_{2})$ whose extension is a suffix of $\alpha$.
Since $r_{4}$ is not a substring of $r_{3}$, this fact contradicts the
assumption that $r_{3}$  maximizes $|\delta|$.
Consequently $(r_{3}, r_{2})$  is irreducible.

Moreover, let $ARC(\alpha,
\alpha\beta\delta, X)$ be a generic  correct
arc-set (at least one such correct arc-set exists, when
$X=R^{P}(\beta\delta)$).
Since $(r_{3}, r_{2})$ is correct, then $r_{2}\in X$ hence  $ARC(\alpha,
\alpha\beta\delta, X)$ is nonempty.
\end{proof}

A direct consequence of Lemma~\ref{lemma:reduction-by-terminal} is
that a nonempty correct terminal arc-set $ARC(\alpha,
\alpha\beta\delta, X)$ implies that  all arcs of the form
$(\gamma\alpha\beta, \beta\delta\eta)$, with $\gamma,\eta\neq\epsilon$ 
are transitive.

Algorithm~\ref{algorithm:sketch-reduction} classifies the
arcs of the overlap graph into reducible or irreducible 
by iteratively computing arc-sets of
increasing extension length $\ell$,  starting from the basic arc-sets $ARC(\epsilon, \epsilon\beta, R^p(\beta))$ of extension length $\ell=0$ obtained in the previous phase.    
By Lemma~\ref{lemma:reduction-by-terminal}, we compute all
correct
arc-sets 
of a given extension length and we discard (according to Definition~\ref{definition:reducing-arc-set}) all arcs of non-terminal arc-sets
that are \emph{removed} by terminal arc-sets.
The set $D$ is used to store the reads of the destination sets $X$ of the computed terminal arc-sets.
Notice that if $ARC(\alpha,\alpha\beta,X)$ is terminal, then all of its
arcs have the same origin $r=\alpha\beta$, that is
$ARC(\alpha,\alpha\beta,X) = \{(r=\alpha\beta, \beta\gamma): \beta\gamma\in X\}$.

The processed arc-sets (that have the same extension length) are partioned into clusters $C(\cdot)$.
We denote with $C(\alpha)$ the set of the arc-sets $ARC(\alpha,\cdot,\cdot)$ with a given extension $\alpha$ which are contained in the stack \emph{Clusters} at a certain point of the execution of Algorithm~\ref{algorithm:sketch-reduction}.
Since the arc-sets pushed to \emph{Clusters}
(lines~\ref{alg:0:step2:end},~\ref{alg:0:final:push}) have an extension $c\alpha$, then for each $\alpha$
there can be at most a cluster $C(\alpha)$ during the entire execution of Algorithm~\ref{algorithm:sketch-reduction}. Observe that at the first iteration (when the processed arc-sets are \emph{basic} and have the same extension $\epsilon$) there is only one cluster $C(\epsilon)$ that is the output of Algorithm~\ref{algorithm:sketch-overlap}.

Each cluster $C(\alpha)$ is processed independently from the other ones, and the set $D$ is used to store the reads of the destination sets $X$ of its terminal arc-sets $ARC(\alpha, \alpha\beta, X)$.
A consequence of Lemma~\ref{lemma:reduction-by-terminal} is that, for each one of the computed terminal arc-sets $ARC(\alpha, \alpha\beta, X)$, all
arcs in $C(\alpha)$ with a destination in $X \in D$ and with an origin different
from $r=\alpha\beta$ are transitive and can be removed simply by removing $X$
from all destination sets in the non-terminal arc-sets of $C(\alpha)$.
Another application of Lemma~\ref{lemma:reduction-by-terminal} is
that,  when we find a terminal arc-set, then all of its arcs
are irreducible, that is it is also correct.
In fact by Lemma~\ref{lemma:reduction-by-terminal}, an arc $(\alpha^{*}\alpha\beta^*, \beta \delta)$, where $\beta^*$ (overlap of the arc) is a prefix of $\beta$, is classified as transitive in relation to the existence of a read $r=\alpha\beta$ that is the origin of an arc $(\alpha\beta, \beta\delta)$ with (left) extension $\alpha$. 
Since the algorithm considers arc-sets by increasing extension length, all arcs that
have extensions shorter than $|\alpha|$ have been reduced in a previous step of the
algorithm  and thus terminal arc-sets computed by previous iterations contain only irreducible arcs.
More precisely, the test at line~\ref{algorithm:test-terminal} is true
iff the current arc-set is terminal.    
In that case, at line~\ref{algorithm:output-arc} all arcs of the arc-set are output as arcs of the string graph, and at
line~\ref{algorithm:accumulate-reducible} the reads in the destination set $X$ is
added to the set $D$ that contains the destinations of $C(\alpha)$ that must be removed from the destination sets of non-terminal arc-sets.

For each cluster $C(\alpha)$, we read twice all arc-sets that are
included in $C(\alpha)$.    
The first time to determine which arc-sets are
terminal and, in that case, to determine the reads (see the set $D$) that
must be removed from all destinations of the non-terminal arc-sets included in $C(\alpha)$.
The second time to compute, from the non-terminal arc-sets $ARC(\alpha, \beta, X)$, the clusters
$C(c\alpha)$, for $c \in \Sigma$, that will contain the nonempty arc-sets $ARC(c\alpha, \beta, X \setminus D)$ with
extension $c\alpha$ consisting of the arcs that 
we still have to check if they are transitive or not. 
Notice that, in Algorithm~\ref{algorithm:sketch-reduction}, the cluster
$C(\alpha)$ that is currently analyzed is stored in \emph{CurrentCluster},
that is a list of the arc-sets included in the cluster.  Terminal arc-sets are removed from \emph{CurrentCluster} before computing the extended clusters $C(c\alpha)$ (see line~\ref{algorithm:remove-terminal-from-cluster}).
Moreover, the clusters that still have to be analyzed are stored in
the stack \emph{Clusters}.    
We use a stack to guarantee that the clusters are analyzed in the
correct order, that is the cluster $C(\alpha)$ is analyzed after all
the clusters $C(\alpha[i:])$ where $\alpha[i:]$ is a generic suffix of $\alpha$ --- 
Lemma~\ref{lemma:produced-arcs} will show that a generic irreducible arc $(r_1, r_2)$ with extension $\alpha$ 
and overlap $\beta$  belongs exactly to the clusters
$C(\epsilon), \ldots, C(\alpha[3:]), C(\alpha[2:]), C(\alpha) $.    
Moreover, $r_2$ does not belong to the set $D$ when
considering $C(\epsilon), \ldots, C(\alpha[3:]), C(\alpha[2:])$, hence the arc  $(r_1,
r_2)$ is correctly output by the algorithm when considering the cluster $C(\alpha)$.



Each cluster $C(\alpha)$ is analyzed separately, and each arc-set in any given
cluster is tested to determine if the arc-set is terminal.
In that case (see the test condition at line~\ref{algorithm:test-terminal}), the arcs having origin in $r=\alpha \beta$ and destination in a read of $X$ are produced in output (see line~\ref{algorithm:output-arc}).
All such arcs  have label $\alpha$.  Moreover the reads in the destination set $X$ are  added to the set $D$, initially empty for
each cluster, that will contain  the destination  sets of all  the terminal
arc-sets computed at the current iteration. We will use this information in the
next step in order to remove all the arcs that have $\alpha$ as suffix of
the label.

After having analyzed all arc-sets in $C(\alpha)$, those arc-sets are
scanned again.
During this second scan, for each arc-set with 
overlap $\beta$ and destination set $X$, we test if there is at least a
read in $X \setminus D$, that is at least a read of the destination set has
not been reduced.
In that case, we split $C(\alpha)$ into the
non-empty  $c\alpha$-cluster (recall that such a cluster is empty iff
$\suff{c\alpha\beta} = 0$).

We can now prove that all irreducible arcs are actually output by our algorithm.

\begin{lemma}
\label{lemma:produced-arcs}
Let $e_1$ be an irreducible arc $(r_1, r_2)$ with extension $\alpha$
and overlap $\beta$.
Then $e_1$ belongs exactly to the  $|\alpha| +1$ clusters
  $C(\alpha), C(\alpha[2:]), C(\alpha[3:]), \ldots, C(\epsilon)$,
while $r_2$ does not belong to the set $D$ when
\emph{currentCluster} is any of $C(\alpha[2:]), C(\alpha[3:]), \ldots,
C(\epsilon)$.
Moreover, $e_1$   is output by the algorithm when \emph{currentCluster} is $C(\alpha)$.
\end{lemma}

\begin{proof}
By construction, $e_1$ can belong only to the clusters
$C(\alpha), C(\alpha[2:]), C(\alpha[3:]), \ldots, C(\epsilon)$.

Now we will prove that $e_{1}$ belongs to all clusters
$C(\alpha), C(\alpha[2:]), C(\alpha[3:]), \ldots, C(\epsilon)$,
while $r_2$ does not belong to the set $D$ when
\emph{currentCluster} is any of $C(\alpha[2:]), C(\alpha[3:]), \ldots,
C(\epsilon)$.
Notice that $e_{1}\in C(\epsilon)$.
Assume to the contrary that there exists $i \ge 2$ such that
$e_{1}\in C(\alpha[i:])$ and $r_2\in D$ when considering a cluster
$C(\alpha[i:])$.
Since $r_2\in D$, by Lemma~\ref{lemma:reduction-by-terminal} there
exists a nonempty terminal arc-set $ARC(\alpha[i:], \alpha[i:]\beta\gamma, X)$ s.t. $r_2 = \beta\gamma\delta$
and $r_2 \in X$.
Since it is terminal and nonempty, such arc-set contains the arc
$(\alpha[i:]\beta\gamma, r_2)$ with extension $\alpha[i:]$.
Since $\alpha[i:]$ is a suffix of $\alpha$ the arc $e_1$ is
transitive, which is a contradiction.

In particular, when the algorithm examines $C(\alpha[2:])$, then  $e_{1}\in C(\alpha[2:])$ and $r_{2}\in X\setminus D$.    
Moreover, $e_{1}$ belongs to the arc-set $ARC(\alpha, \alpha\beta,
X\setminus D)$ added to \emph{ExtendedClusters}$[\alpha[1]]$ at line~\ref{alg:0:new-triples}.    
Clearly, such arc-set is included in $C(\alpha)$.    
When the algorithm examines the cluster $C(\alpha)$, the arc-set containing
$e_{1}$ satisfies the condition at line~\ref{algorithm:test-terminal},
hence such arc is output.
\end{proof}

\begin{corollary}
\label{corollary:produced-arcs}
  The set of arcs computed by the algorithm is a superset of the irreducible arcs of the string graph.
\end{corollary}

\begin{lemma}
\label{lemma:all-triples-are-correct}
Let $ARC(\alpha,\alpha\beta,X)$ be an arc-set inserted into a cluster by
Algorithm~\ref{algorithm:sketch-reduction}. 
Then such an arc-set is correct.
\end{lemma}

\begin{proof}
Let $e_1$ be an irreducible arc $(r_1, r_2)$ of
$ARC(\alpha,\alpha\beta,X)$, and let $\alpha_{1}$ be
respectively the extension and the overlap $\beta$ of $e_{1}$.
Since $e_{1}\in ARC(\alpha,\alpha\beta,X)$, then $\alpha$ is a suffix of
$\alpha_{1}$, therefore we can apply Lemma~\ref{lemma:produced-arcs} which
implies that $e_{1}\in C(\alpha)$.
Since the only arc-set contained in $C(\alpha)$ to which
$e_{1}$ can belong is $ARC(\alpha,\alpha\beta,X)$, then $r_{2}\in X$
which completes the proof.
\end{proof}

We can now prove that no transitive arc is ever output.


\begin{lemma}
\label{lemma:reduced-arcs}
Let $e_1$ be a transitive arc $(r_1, r_2)$ with overlap $\beta$.
Then the algorithm does not output $e_1$.
\end{lemma}

\begin{proof}
Since $e_{1}$ is transitive,  by
Lemma~\ref{lemma:reduction-by-terminal} the two reads $r_{1}$, $r_{2}$ are
$r_1 = \gamma\alpha\beta$,
$r_2 = \beta\delta\eta$, and
there exists an input string $r_{3}=\alpha\beta\delta$ such that the arc
$e_{2}=(r_3, r_2)$ with overlap $\beta\delta$ is irreducible.
Moreover all
correct arc-sets of the form $ARC(\alpha, \alpha\beta\delta, X)$ with $r_{2}\in X$ are
nonempty and terminal.

Assume to the contrary that $e_{1}$ is output by
Algorithm~\ref{algorithm:sketch-reduction}, and notice that such arc can be output
only when the current cluster is $C(\alpha)$ and the current arc-set is $ARC(\gamma\alpha,\gamma\alpha\beta, X)$ with $r_{2}\in X$.

By the construction of our algorithm,
since the cluster $C(\gamma\alpha)$ is nonempty, also $C(\alpha)$ is
nonempty: let us consider the iteration when the current cluster is
$C(\alpha)$.
By Lemma~\ref{lemma:all-triples-are-correct} the arc-set
$ARC(\alpha, \alpha\beta\delta, X_{1})$ is correct, hence it contains
the arc $e_{2}$.
But such arc-set satisfies the condition at
line~\ref{algorithm:test-terminal}, hence $r_{2}\in D$ at that
iteration.
Consequently, $C(\alpha)$ cannot contain an arc-set with
destination set with $r_{2}$.
\end{proof}



Theorem~\ref{theorem:correctness} is a direct consequence of
Corollary~\ref{corollary:produced-arcs} and 
Lemma~\ref{lemma:reduced-arcs}.

\begin{theorem}
\label{theorem:correctness}
  Given as input a set of strings $R$, Algorithm~\ref{algorithm:sketch-reduction} computes exactly
  the arcs of the string graph.
\end{theorem}


\begin{algorithm}[tb!]
\SetKwInOut{PRE}{Input}
\SetKwInOut{POST}{Output}
\PRE{The set $R$ of input strings}
\POST{The set $Basics$ of the basic arc-sets $ARC(\epsilon, \epsilon\beta, R^p(\beta))$}
$Basics \gets$ empty list\;

$Last \gets \{ c \in \Sigma \mid \suff{c} > 0 \text{ and } \substr{c} > \suff{c}\}$\label{alg:0:symbols}\;
\While{Last $\neq\emptyset$}{\label{alg:0:step1:begin}
    $New \gets \emptyset$\;
  \ForEach{$\beta^{*}\in \text{Last}$}{%
    \If{$\pref{\beta^{*}} > 0$}{%
      Append $(\epsilon,\epsilon\beta^{*},\listpref{\beta^{*}})$ to $Basics$\label{alg:0:listpref}\;
    }
    \For{$c \in \Sigma$\label{alg:0:begin-prepend}}{%
      \If{$\suff{c\beta^{*}} > 0$ and $\substr{c\beta^{*}} > \suff{c\beta^{*}}$}{%
        Add $c\beta^{*}$ to $New$\label{alg:0:end-prepend}\;
      }
    }
  }
  $Last \gets New$\label{alg:0:step1:end}\;
}
\Return $Basics$\;
  \caption{Compute the overlap graph}\label{algorithm:sketch-overlap}
\end{algorithm}

\begin{algorithm}[tb!]
\SetKwInOut{PRE}{Input}
\SetKwInOut{POST}{Output}
\PRE{The set $Basics$ of the basic arc-sets $ARC(\epsilon, \epsilon\beta, R^p(\beta))$}
\POST{The arcs of the string graph of $R$}
$Clusters \gets$ empty stack\;
Push $Basics$ to $Clusters$\;

\While{Clusters is not empty}{\label{alg:0:step2:begin}
  $CurrentCluster \gets$ Pop($Clusters$)\;
  $D \gets \emptyset$\;
  $ExtendedClusters \gets$ an array of $|\Sigma|$ empty clusters\;
  \ForEach{$(\alpha, \alpha\beta, X)\in$ CurrentCluster\label{algorithm:step1:end}}{%
    \If{$\substr{\alpha\beta} = \pref{\alpha\beta} = \suff{\alpha\beta} > 0$\label{algorithm:test-terminal}}{%
    	\ForEach{$x \in X$}{%
    		Output the arc $(\alpha\beta, x)$ with label $\alpha$\label{algorithm:output-arc}\;
    	}
      $D \gets D \cup X$\label{algorithm:accumulate-reducible}\;
      Remove $(\alpha, \alpha\beta, X)$ from \emph{CurrentCluster}\label{algorithm:remove-terminal-from-cluster}\; 
    }
  }
  \ForEach{$(\alpha, \alpha\beta, X)\in$ CurrentCluster\label{alg:0:2for:begin}}{%
      \ForEach{$c \in \Sigma$}{%
        \If{$\suff{c\alpha\beta} > 0$\label{algorithm:test-nonempty-triple}}{%
          Append $(c\alpha, c\alpha\beta, X \setminus D)$ to $ExtendedClusters[c]$\label{alg:0:new-triples}\label{alg:0:step2:end}\;
        }
    }
  }
  \ForEach{$c \in \Sigma$}{%
  	\If{$ExtendedClusters[c] \ne \emptyset$}{%
  	  Push $ExtendedClusters[c]$ to \emph{Clusters}\label{alg:0:final:push}\;
  	}
  }
}
  \caption{Compute the string graph}\label{algorithm:sketch-reduction}
\end{algorithm}

\section{Computational Complexity and Data representation}
\label{sec:data-representation}

We can now study the time complexity of our algorithm.
Previously, we have shown that Algorithm~\ref{algorithm:sketch-overlap} produces at most
$O(nm)$ basic arc-sets, one for each distinct overlap $\beta$.
Moreover, notice that Algorithm~\ref{algorithm:sketch-overlap} requires constant time for
each potential overlap.
Since each potential overlap is a suffix of a read, there are $O(nm)$ potential overlaps,
hence the time complexity of Algorithm~\ref{algorithm:sketch-overlap} is $O(nm)$.

A similar argument shows that the number of arc-sets managed by
Algorithm~\ref{algorithm:sketch-reduction} is at most $O(nm^{2})$, since each suffix
$\alpha\beta$ can be considered for different extensions $\alpha$.
Moreover, differently from  Algorithm~\ref{algorithm:sketch-overlap}, the time spent by
Algorithm~\ref{algorithm:sketch-reduction} for each string  $\alpha\beta$ is not constant.
More precisely, for each cluster, besides computing \substr{\cdot}, \pref{\cdot},
\suff{\cdot}, Algorithm~\ref{algorithm:sketch-reduction} computes a union $D$
(line~\ref{algorithm:accumulate-reducible}) and the difference $X\setminus D$
(line~\ref{alg:0:new-triples}) --- a direct inspection of the pseudocode shows that all
other operations require constant time for each string $\alpha\beta$.

The union at 
line~\ref{algorithm:accumulate-reducible} and the difference  $X
\setminus D$ at line~\ref{alg:0:new-triples} are computed for each string suffix
$\alpha\beta$ where $\beta$ is a potential overlap.
Let $d(n)$ be the time complexity of those two operations on
$n$-element sets (the actual time complexity depends on the data
structure used).    
Therefore, the time complexity of the entire algorithm is $O(nmd(n))$.
We point out that a representation based on an $n$-long bitvector (like the one we will
discuss in the following) implies an $O(n)$ time complexity for each union and difference.

As a consequence, the time complexity of our algorithm is $O(nm(n+m))$.
On the other hand, we conjecture that this time complexity is highly pessimistic, since
usually the number of  potential overlaps is smaller than the worst case.

One of the main features of our approach is that we operates only on the (potentially
compressed) FM-index of the collection of input reads.
To achieve that goal we cannot use the na\"{\i}ve representation of a string $\omega$,
but we must employ a BWT-based representation.
More precisely, we represent a string $\omega$ with the $\omega$-interval $\q{\omega} =
[b_{\omega}, e_{\omega}]$  and the fact that $\omega$ is suffix in some read of $R$ with the  $\omega\$$-interval $\q{\omega\$} = [b_{\omega\$},
e_{\omega\$}]$ on the BWT, hence using four integers for each string.

We need to show how to compute efficiently this representation of $\omega$ as well as $\pref{\omega}$, $\suff{\omega}$, and $\substr{\omega}$.
The following proposition is instrumental and can be verified by a direct inspection of
Algorithms~\ref{algorithm:sketch-overlap} and~\ref{algorithm:sketch-reduction}: both
algorithms compute strings by prepending characters, that is they need to obtain the
(representation of the) string $c\omega$ from the  (representation of the) string $\omega$
that has been processed previously.

\begin{proposition}
\label{proposition:algorithm-via-extensions}
Let $\omega$ be 
a string processed by Algorithms~\ref{algorithm:sketch-overlap}
or~\ref{algorithm:sketch-reduction}, and such that $|\omega| > 1$.
Then the string $\omega[2:]$ is processed by the same algorithm before $\omega$.
\end{proposition}

Since each string $\omega$ considered by the algorithm is a substring
of some input read, we can represent $\omega$ in constant space by the
boundaries (\ie the first and the last index) of $\q{\omega}$ and $\q{\omega\$}$, 
instead of using the naïve representation with $O(|\omega|)$ space.
Consequently, the algorithm  operates only on
the (potentially compressed) FM-index of the collection of input reads.

Initially we compute the intervals $\q{c}$ and $\q{c\$}$ for each character $c\in \Sigma$ by
inspecting the FM-index, hence settling the case when $|\omega| = 1$.
Furthermore, the FM-index allows to compute in $O(1)$ time $\q{c\omega}$ and
$\q{c\omega\$}$ from $\q{\omega}$ and $\q{\omega\$}$ --- the backward $c$-extension of
$\omega$ and $\omega\$$ ---
for any character $c\in \Sigma \cup \{\$\}$~\citep{Ferragina2005}.

This representation allows to answer each query
$\pref{\omega}$, $\suff{\omega}$ and $\substr{\omega}$ in $O(1)$ time.
In fact, given $\q{\omega} = [b_\omega, e_\omega]$, then $\substr{\omega}=e_\omega - b_\omega + 1$
and $\pref{\omega}=e_{\$\omega} - b_{\$\omega} + 1$ where
$\q{\$\omega}=[b_{\$\omega},e_{\$\omega}]$ is the result of the backward $\$$-extension of
$\q{\omega}$ --- which can be obtained in $O(1)$ time.

Moreover, $\listpref{\omega}$ 
corresponds to the set of reads appearing in the interval
$\q{\$\omega}$ of the GSA\@.
Notice that no read can appear twice in such interval, hence a linear scan of the interval
$\q{\$\omega}$ suffices.

Answering to the query $\suff{\omega}$ requires considering the interval $\q{\omega\$} =
[b_{\omega\$}, e_{\omega\$}]$ (which is the reason it is included in the representation of
$\omega$).
In fact $\suff{\omega}=e_{\omega\$} - b_{\omega\$} + 1$.

A further optimization of the representation is possible.
Recall that $\q{\omega} = [b_\omega, e_\omega]$ and $\q{\omega\$} = [b_{\omega\$}, e_{\omega\$}]$.
Notice that $b_{\omega} = b_{\omega\$}$, since the sentinel $\$$ is lexicographically
smaller than all characters in $\Sigma$.
Hence the two intervals $\q{\omega}$ and $\q{\omega\$}$ can be represented by the three
integers $b_\omega, e_\omega, e_{\omega\$}$.

Instead of storing directly the three integers $b_\omega, e_\omega, e_{\omega\$}$, we use
two $n(m+1)$-long bitvectors, requiring $2n(m+1) + o(nm)$ bits and allowing to answer in
constant time to rank and select queries~\citep{clark_efficient_1996,jacobson_space-efficient_1989}.

Algorithm~\ref{algorithm:sketch-overlap} mainly has to represent the set of
potential overlaps (\ie the lists \emph{Last} and \emph{New}).
At each iteration, the potential overlaps in \emph{Last} (in \emph{New}, resp.) have the
same length, hence their corresponding intervals on the BWT are disjoint.
For each potential overlap $\beta \in Last$ (in \emph{New}, resp.) represented by the
triple $(b_{\beta}, e_{\beta}, e_{\beta\$})$, the first 
bitvector has $1$ in position $b_\beta$ and the second bitvector has $1$
in positions $e_{\beta\$}$ and $e_\beta$ --- since each $\q{\beta}$ is disjoint from all other
intervals, the $i$-th interval is represented by the $i$-th $1$ in the first bitvector, and
the $2i$-th and $(2i+1)$-th $1$s of the second bitvector.

Algorithm~\ref{algorithm:sketch-reduction} mainly has to represent clusters and, for each
cluster, the set $D$ containing the reads that must be deleted from all destination set of
the non-terminal arc-sets of the cluster.
A cluster groups together arc-sets whose overlaps are either pairwise different
or one is the prefix of the other.
Thus, the corresponding intervals on the BWT are either disjoint or
one contained in the other
(\ie partial overlap of the intervals cannot happen).

Moreover, also the destination set of the \emph{basic} arc-sets can be
represented by a set of pairwise disjoint or contained intervals on the BWT
(since $\listpref{\beta}$ of line~\ref{alg:0:listpref} corresponds to the reads of the
interval $\q{\$\beta}$ on the GSA).
The following proposition, which can be proved by a similar argument as
Lemma~\ref{lemma:produced-arcs}, describes the relation between destination sets.

\begin{proposition}
\label{proposition:arcset-in-same-cluster}
Let $ARC(\alpha, \alpha\beta_1, X_1)$, $ARC(\alpha, \alpha\beta_2, X_2)$ be two arc-sets
belonging to the cluster $C(\alpha)$ during the
execution of Algorithm~\ref{algorithm:sketch-reduction}.
Then $X_{1}\cap X_{2}\neq\emptyset$, or one of $X_{1}$, $X_{2}$ is a subset of the other.
Moreover, if $\beta_1$ is a prefix of $\beta_2$, then $X_2 \subseteq X_1$.
\end{proposition}

Given a cluster $C(\alpha)$, let $\gamma$ be the longest common prefix of all strings
$\alpha\beta$ such that $ARC(\alpha, \alpha\beta, X)$ is an arc-set of the cluster.
Then the numbers that represent the strings $\alpha\beta$ are all contained in the
$\gamma$-interval: therefore it suffices to consider numbers in the interval  
$[b_{\gamma}, e_{\gamma}]$ instead of the interval $[1, n(m+1)]$.
A more compact representation of the data that we have to manage consists of two
$(e_{\gamma} - b_{\gamma} +1)$-long vectors $V_b, V_e$ of integers, and a bitvector $B_x$
of length $e_{\$\gamma} - b_{\$\gamma} +1$ --- where $[b_{\$ \gamma}, e_{\$\gamma}]$ is
the $\$\gamma$-interval.
Each entry $V_b[i]$ ($V_e[i]$, resp.) is the number of
arc-sets $ARC(\alpha,\alpha\beta, \cdot)$  in $C(\alpha)$ with initial (final, resp.) boundary $b_\gamma +
i$.
Moreover, the destination sets that are considered in the cluster are contained in the set
of reads in $\q{\$\gamma}$, hence the  bitvector $B_{x}$ is to encode the set $D$: in
fact $B_x[i]$ is $1$ iff the $(b_{\$\gamma} + i)$-th read (in lexicographic order),
belongs to $D$.

Notice that the arc-sets $ARC(\alpha,\alpha\beta, X)$ in each cluster $C(\alpha)$ are
sorted according to the lexicographic order of $\beta$: this fact allows to coordinate the
sequential scan  of $V_{b}$ and $V_{e}$ that corresponds to the for loop at
lines~\ref{alg:0:2for:begin}--\ref{alg:0:step2:end} with a scan of the list of
destination sets (one for each arc-set) that must be updated at line~\ref{alg:0:step2:end}.
Also notice that $V_{b}$, $V_{e}$, and $B_{x}$ are constructed at
lines~\ref{algorithm:accumulate-reducible}--\ref{alg:0:step2:end} and exploited at
lines~\ref{alg:0:2for:begin}--\ref{alg:0:step2:end}, hence the data structures necessary
to answer rank and select queries in constant time must be built once for each cluster
just before line~\ref{alg:0:2for:begin}.

\section{Experimental Analysis}
\label{sec:experimental-analysis}

A C++ implementation of our approach, called FSG (short for Fast String Graph), has been
integrated in the SGA 
suite and is available at \url{http://fsg.algolab.eu} under the GPLv3 license.
Our implementation uses the Intel\textsuperscript{\textregistered} Threading
Building Blocks library in order to manage the parallelism.
The software is conceptually divided in the two phases illustrated in the previous
section, and each phase has been implemented as a computational pipeline using
the \texttt{pipeline} construct made available by the library.

We have evaluated the performance of FSG with three experiments: the first experiment
compares FSG and SGA on a standard benchmark of 875 million 101bp-long reads sequenced
from the NA12878 individual of the International HapMap and 1000 genomes project
(extracted from \url{ftp://ftp-trace.ncbi.nih.gov/1000genomes/ftp/technical/working/20101201_cg_NA12878/NA12878.hiseq.wgs.bwa.recal.bam}).
The second experiment, on an \emph{Escherichia coli} dataset, aims at investigating the
cause of the speedup obtained by FSG with respect to SGA.
Finally, the third experiment (which is on a synthetic dataset obtained from the Human
chromosome 1), studies the effect of coverage on the performance of FSG and SGA.
All experiments have been performed on an Ubuntu 14.04 server 
with four 8-core Intel\textsuperscript{\textregistered}
Xeon E5-4610v2 2.30GHz CPUs (hyperthreading was enabled for a total of 16 threads per processor).
The server has a NUMA architecture with 64GiB of RAM for each node (256GiB in
total).
To minimize the effects of the architecture on the executions, we used
\texttt{numactl} to preferably allocate memory on the
first node where also the threads have been executed (with 32
threads also the second node was used).

We have run SGA with its default parameters, that is SGA has computed exact overlaps after
having corrected the input reads.    
We could not compare FSG with Fermi,  since Fermi does not split its steps
in a way that allows to isolate the running time of the string graph
construction---most notably, it includes reads correction and scaffolding.
Since the string graphs computed by FSG and SGA are essentially the same, we have not
focused our analysis on the quality of the resulting assemblies, but we give a brief
analysis hinting that it is not possible to determine which assemblies are better (see Table~\ref{table:full-dataset-total-quality}).

\begin{table}[ht!]
\caption{Quality of the assemblies computed by FSG and SGA.}
\begin{center}
\begin{tabular}{lrr}
\toprule
& \multicolumn{1}{c}{SGA} & \multicolumn{1}{c}{FSG} \\
\otoprule
N. Contigs ($\ge 0$bp) & $15,322,517$ & $14,904,770$  \\
N. Contigs ($\ge 5000$bp) & $136,717$ & $136,693$  \\\midrule
Tot. Contig Length ($\ge 0$bp) & $4,154,574,477$ & $4,111,303,910$ \\
Tot. Contig Length ($\ge 5000$bp) & $1,173,041,496$ & $1,173,000,932$ \\
Tot. Length ($\ge 25000$bp) & $26,665,111$ & $26,674,888$ \\
N50 & $4,700$ & $4,700$ \\
N75 & $2,393$ & $2,393$ \\
\bottomrule
\end{tabular}
\end{center}
\label{table:full-dataset-total-quality}
\end{table}

For genome assembly purposes, only overlaps whose length is
larger than a user-defined threshold are considered.
The value of the minimum overlap length threshold that empirically showed the
best results in terms of genome assembly quality is around the 75\% of the read
length~\citep{Simpson2012}.
In order to assess how graph size affects performance, different values of
minimum overlap length (called $\tau$) between reads have been used (clearly,
the lower this value, the larger the graph).
The minimum overlap lengths used in this experimental assessment are 55, 65,
75, and 85, hence the chosen values test the approaches also on
larger-than-normal ($\tau=55$) and smaller-than-normal ($\tau=85$) string graphs.

Another aspect that we wanted to measure is the scalability of FSG for a different number
of threads.
We have run the programs with 1, 4, 8, 16, and 32 threads.    
In all cases, we have measured the elapsed (wall-clock) time and the total CPU time (the time
a CPU has been working).

\begin{table}[tb!]
\caption{Comparison of FSG and SGA, for different minimum overlap
lengths and numbers of threads.
The wall-clock time is the time used to compute the string graph.    
The CPU time is the overall execution time over all CPUs actually
used.
}
\label{tab:times}
\footnotesize\centering
\begin{tabular}{llllllll}
\toprule
 Min.    &no.~of & \multicolumn{3}{c}{Wall time [min]} & \multicolumn{3}{c}{Work time [min]}\\
overlap	&threads& FSG	& SGA	& $\frac{\text{FSG}}{\text{SGA}}$ & FSG	& SGA & $\frac{\text{FSG}}{\text{SGA}}$\\
\otoprule
55	& 1	& 1,485	& 4,486	& 0.331 & 1,483	& 4,480& 0.331\\
	& 4	& 474	& 1,961	& 0.242	& 1,828	& 4,673& 0.391\\
	& 8	& 318	& 1,527	& 0.209	& 2,203	& 4,936& 0.446\\
	& 16	& 278	& 1,295	& 0.215	& 3,430	& 5,915& 0.580\\
	& 32	& 328	& 1,007	& 0.326	& 7,094	& 5,881& 1.206\\
\midrule								    
65	& 1	& 1,174	& 3,238	& 0.363	& 1,171	& 3,234& 0.363\\
	& 4	& 416	& 1,165	& 0.358	& 1,606	& 3,392& 0.473\\
	& 8	& 271	& 863	& 0.315	& 1,842	& 3,596& 0.512\\
	& 16	& 255	& 729	& 0.351	& 3,091	& 4,469& 0.692\\
	& 32	& 316	& 579	& 0.546	& 6,690	& 4,444& 1.505\\
\midrule								    
75	& 1	& 1,065	& 2,877	& 0.37	& 1,063	& 2,868& 0.371\\
	& 4	& 379	& 915	& 0.415	& 1,473	& 2,903& 0.507\\
	& 8	& 251	& 748	& 0.336	& 1,708	& 3,232& 0.528\\
	& 16	& 246	& 561	& 0.439	& 2,890	& 3,975& 0.727\\
	& 32	& 306	& 455	& 0.674	& 6,368	& 4,062& 1.568\\
\midrule								    
85	& 1	& 1,000	& 2,592	& 0.386	& 999	& 2,588& 0.386\\
	& 4	& 360	& 833	& 0.432	& 1,392	& 2,715& 0.513\\
	& 8	& 238	& 623	& 0.383	& 1,595	& 3,053& 0.523\\
	& 16	& 229	& 502	& 0.457	& 2,686	& 3,653& 0.735\\
	& 32	& 298	& 407	& 0.733	& 6,117	& 3,735& 1.638\\
\bottomrule				   
\end{tabular}
\end{table}

In terms of memory, SGA does not maintain the computed string graph in memory,
hence its peak memory usage is only dependent on the input size and in these
experiments was always about 63GiB.
Also the peak memory usage of our approach was
approximately equal to 138GiB for all the configurations.
As a consequence, the memory usage of our approach is practically only dependent
on the input size since it compactly stores the arc-sets of the first phase
and the stack maintaining the clusters to be processed in the second
phase does not grow to have more than $|\Sigma|\cdot(m - \tau)$
elements.

Table~\ref{tab:times} summarizes the running times of both approaches on the
different configurations of the parameters.
Notice that FSG approach is from 2.3 to 4.8 times faster
than SGA in terms of wall-clock time and from 1.9 to 3 times in terms of CPU time.
On the other hand, FSG uses approximately 2.2 times the memory used by SGA --- on the
executions with at most 8 threads.    

While FSG is noticeably faster than SGA on all instances, there are
some other interesting observations.    
The combined analysis of the CPU time and the wall-clock time on at most 8
threads (which is the number of physical cores of each CPU on our server)
suggests that FSG is more CPU efficient than SGA and is able to better
distribute the workload across the threads.
The latter value of 8 threads seems to be a sweet spot for the
parallel version of FSG.    

On a larger number of threads, and in particular the fact that the elapsed time
of FSG on 32 threads is larger than that on 16 threads suggests that, in its
current form, FSG might not be suitable for a large number of threads.    
However, since the current implementation of FSG is almost a proof of concept,
future improvements to its codebase and a better analysis of the race conditions
of our tool will likely lead to better performances with a large number of
threads.
Furthermore, notice that also the SGA algorithm, which is (almost) embarrassingly
parallel 
and has a stable implementation, does not achieve a speed-up better than 6.4
with 32 threads.
As such, a factor that likely contributes to a poor scaling behaviour of both
FSG and SGA could be also the NUMA architecture of the server used
for the experimental analysis, which makes different-unit memory
accesses more expensive (in our case, the processors in each unit can
manage at most 16 logical threads, and only 8 on physical cores).

Notice that, FSG uses more memory than SGA. 
The reason is that genome assemblers have to correctly manage reads extracted from both
strands of the genome.    
In our case, this fact has been addressed by adding each reverse-and-complement read to
the set of strings on which the FM-index has been built, hence immediately doubling the
size of the FM-index.    
Moreover, FSG needs some additional data structures to correctly maintain potential
overlaps and arc-sets, as described in Section~\ref{sec:data-representation}.
The main goal of FSG is to improve the running time, and not necessarily to
decrease memory usage.

We also wanted to estimate the effect of the optimizations discussed in
Section~\ref{sec:data-representation}, by measuring the number of backward extensions performed.
We have instrumented FSG and SGA to count the number of accesses to the FM-index (which
can happen only when computing a backward extension) and we have run both program on the
NA12878 dataset with $\tau=85$: FSG has made $877 \cdot 10^9$ accesses, while SGA has made
$947 \cdot 10^9$ accesses, \ie SGA has made 8\% more backward extensions than FSG.
This result confirms that clustering together arc-sets and avoiding unnecessary backward
extensions actually improves on SGA's strategy.
On the other hand, such difference on the number of backward extensions is unable to fully
explain the different running times of FSG and SGA --- Table~\ref{tab:times} shows that,
on this specific instance, SGA needs 2.5x the time used by SGA.
Therefore we have designed a second experiment to investigate the main causes of the
different performances of FSG and SGA.

In our second experiment, we have run both FSG and SGA on the \emph{Escherichia coli} dataset downloaded
from \url{http://www.ebi.ac.uk/ena/data/view/CP009789} under
valgrind~\citep{DBLP:conf/pldi/NethercoteS07} to measure the instruction and memory access
patterns.
Since running a program under valgrind increases the running time by two orders of
magnitude, it was not feasible to use the same dataset as the first experiment, but we had
to use a much smaller one.
The results of this experiment are summarized in Table~\ref{tab:valgrind-fsg-sga}.    

Our initial conjecture was that the better efficiency achieved by FSG originated from
operating only on the 
FM-index of the input reads and by the order on which extension operations (\ie
considering a new string $c\alpha$ after $\alpha$ has been processed) are
performed.
These two characteristics of our algorithm allow to eliminate the redundant
queries to the index which, instead, are performed by SGA.
In fact, FSG considers each string that is longer than the threshold at most once,
while SGA potentially reconsiders the same string once for each read in which
the string occurs.
A consequence of our conjecture should have been fewer memory accesses and cache misses.    
Unfortunately the results we have obtained on memory accesses are inconclusive: the number
of memory accesses made by SGA are almost twice as many as those made by FSG.    
On the other hand, the number of cache misses that result in RAM accesses hints that SGA
is much more efficient in that regard.    
We have estimated the total time spent in memory accesses using the values suggested by
Valgrind: 1 cycle for each cache access, 10 cycles for each level 1 cache miss, and 100
cycles for each cache miss that results in a RAM access.    
Overall, FSG is more efficient than SGA, but definitely less than 10\% more efficient:
hence this reason alone cannot justify the difference in running times of the two programs.    

The fact that FSG consists of several linear scans suggests that it should be able to
better exploit the superscalar features of modern CPUs.    
In fact, our analysis of branch mispredictions confirms this fact:
the ratio of the  total number of branch predictions made by FSG is about
75\% of those made by SGA, and  the  total number of branch mispredictions made by
FSG is less than 50\% of those made by SGA.

\begin{table}[htb]
\centering
\begin{tabular}{lrrr}
&SGA&FSG&SGA/FSG\\\hline
Instruction read (millions)&1,271,545&743,807&1.710\\
Instruction cache level 1 miss read (millions)&2,187&15&142.642\\
Instruction RAM miss read (millions)&0.048&0.034&1.412\\
\midrule								    
Data read (millions)&358,461&199,845&1.794\\
D level 1 miss read (millions)&13,491&6,206&2.174\\
Data RAM miss read (millions)&146&725&0.202\\
Data write (millions)&143&90&1.589\\
D level 1 miss write (millions)&647&562&1.151\\
Data RAM miss write (millions)&7&273&0.028\\
Total time for memory accesses (millions of cycles)&180,446&168,706&1.070\\
\midrule								    
Conditional branches executed (millions)&159,912&106,167&1.506\\
Conditional branches mispredicted (millions)&10,292&4,989&2.063\\
Indirect branches executed (millions)&10,044&4,977&2.018\\
Indirect branches mispredicted (millions)&1,472&317&4.644\\
\end{tabular}
\caption{Comparing FSG and SGA on the Escherichia coli dataset using valgrind}
\label{tab:valgrind-fsg-sga}
\end{table}

We have compared the string graphs produced by FSG and SGA, since their respective notions
of string graphs are slightly different.
More precisely, FSG keeps multiple arcs with distinct labels between the same pairs of
vertices, while SGA retains only the arc having the shortest label.
When the minimum overlap length is 65bp, the string graph computed by FSG had $\sim 3.5$\% more
arcs than the one computed by SGA, but the impact on the actual assembly is not relevant
(see Table~\ref{table:full-dataset-total-quality}).

\begin{table}[htb]
\centering
\begin{tabular}{p{3em}p{3em}llllllll}
&&&\multicolumn{7}{c}{Coverage}\\
Read Length&Num. threads&program&4&8&16&32&64&128&256\\
101&8&FSG&2&3&6&11&22&36&36\\
101&8&SGA&1&2&3&6&13&17&17\\\cline{2-10}
101&16&FSG&2&3&6&12&22&36&36\\
101&16&SGA&1&2&4&7&13&18&18\\\cline{2-10}
101&32&FSG&3&4&7&13&23&36&36\\
101&32&SGA&2&3&5&8&14&19&19\\\cline{2-10}
101&64&FSG&5&7&9&15&26&38&38\\
101&64&SGA&5&5&7&10&16&21&21\\\midrule
250&8&FSG&1&3&5&10&18&39&78\\
250&8&SGA&1&1&1&3&5&10&19\\\cline{2-10}
250&16&FSG&2&3&6&10&19&40&78\\
250&16&SGA&1&1&2&3&6&11&20\\\cline{2-10}
250&32&FSG&3&4&7&11&20&41&78\\
250&32&SGA&2&3&3&4&7&12&21\\\cline{2-10}
250&64&FSG&5&6&9&14&22&43&80\\
250&64&SGA&5&5&5&7&9&14&23\\
\end{tabular}
\caption{Comparing FSG and SGA on different coverage values and different number of
  threads: Peak memory usage (in Gigabytes)}
\label{tab:scalability-coverage-memory}
\end{table}

\begin{table}[htb]
\centering
\begin{tabular}{p{3em}p{3em}llllllll}
&&&\multicolumn{7}{c}{Coverage}\\
Read Length&Num. threads&program&4&8&16&32&64&128&256\\
101&8&FSG&4&7&13&19&38&59&59\\
101&8&SGA&5&10&20&40&89&148&145\\\cline{2-10}
101&16&FSG&3&6&10&17&34&55&57\\
101&16&SGA&4&8&16&32&68&116&115\\\cline{2-10}
101&32&FSG&4&8&12&21&41&64&65\\
101&32&SGA&4&8&16&32&68&111&120\\\cline{2-10}
101&64&FSG&7&24&44&85&174&265&286\\
101&64&SGA&3&7&14&27&57&92&92\\\midrule
250&8&FSG&8&10&15&24&37&77&171\\
250&8&SGA&4&9&20&43&86&191&416\\\cline{2-10}
250&16&FSG&5&8&11&18&30&64&139\\
250&16&SGA&3&7&15&32&68&144&309\\\cline{2-10}
250&32&FSG&7&11&17&24&37&72&141\\
250&32&SGA&3&7&14&29&62&135&298\\\cline{2-10}
250&64&FSG&5&9&19&42&90&184&421\\
250&64&SGA&2&5&10&22&46&95&209\\
\end{tabular}
\caption{Comparing FSG and SGA on different coverage values and different number of threads: Running times (in hours)}
\label{tab:scalability-coverage-time}
\end{table}

A third experiment has been performed with the goal of studying the scalability of FSG on
larger coverage values.
We have extracted random reads from the human chromosome 1, with different
average coverage values and different read lengths.
Since the goal of this experiment is analyzing the running time and memory usage, while we
are not interested into the accuracy of the predictions, the reads do not contain errors
(\ie they are substrings extracted from the reference genome).
The results of these experiments are reported in
Tables~\ref{tab:scalability-coverage-memory} and~\ref{tab:scalability-coverage-time}.
The comparison between FSG and SGA confirms the results of the first experiment, and shows
that for large coverage values, FSG becomes even faster than SGA but uses an increasing
amount of memory.

\section{Conclusions and future work}

We present FSG: a tool implementing a new algorithm for constructing a string graph that works
directly querying  a FM-index representing a collection of reads, instead of processing
the input reads.    
Our main goal is to provide a simpler and fast algorithm to  construct string graphs, so
that its implementation can be easily integrated into an assembly pipeline that analyzes
the paths of the string graph to produce the final assembly.
Indeed, FSG could be used for related  purposes, such as transcriptome
assembly~\citep{Lacroix2008,Beretta2013}, and haplotype
assembly~\citep{BonizzoniVDL03}, and variant detection via aligning paths of the string
graph against a reference genome. 
These topics are some of the research directions that we plan to investigate.

More precisely, our  algorithm uses string queries that are efficiently implemented using the
information provided by the index and takes advantage  of a lexicographic  based-ordering of string
queries  that allows  to reduce the total number of such queries to build the string graph. Since
FSG  reduces the total number of queries and does not process the input  to compute  the transitive
reduction as done by SGA,  the current state-of-art tool for computing the string graph,  we are
able to show that FSG is significantly faster than SGA over genomic data. It would be interesting to
test our implementation  to compute  the string graph for a large collection of strings with
different characteristics than  genomic reads, such as for example when the  alphabet is  of larger size or  the input data consists of strings  of variable length.

\section*{Acknowledgments}

The authors acknowledge the support of the MIUR PRIN 2010-2011 grant
``Automi e Linguaggi Formali: Aspetti Matematici e Applicativi''  code 2010LYA9RH,
of the Cariplo Foundation grant 2013-0955 (Modulation of anti cancer immune
response by regulatory non-coding RNAs), of the FA 2013 grant ``Metodi
algoritmici e modelli: aspetti teorici e applicazioni in
bioinformatica'' code 2013-ATE-0281, and of the FA 2014 grant
``Algoritmi e modelli computazionali: aspetti teorici e applicazioni
nelle scienze della vita'' code 2014-ATE-0382.


\end{document}